\documentclass[submission,copyright,creativecommons]{eptcs}

\providecommand{\event}{QPL2014}

\usepackage{algpseudocode}
\usepackage{breakurl}
\usepackage{hyperref}
\usepackage{paralist}
\usepackage{etoolbox}
\usepackage{graphicx}
\usepackage{mdwlist}
\usepackage{amsmath}
\usepackage{wrapfig}
\usepackage{picins}
\usepackage{amsthm}
\usepackage{tikz}

\makeatletter
\newcommand*{\algrule}[1][\algorithmicindent]{\makebox[#1][l]{\hspace*{.5em}\vrule height .75\baselineskip depth .25\baselineskip}}%

\newcount\ALG@printindent@tempcnta
\def\ALG@printindent{%
    \ifnum \theALG@nested>0%
        \ifx\ALG@text\ALG@x@notext%
            \addvspace{-3pt}%
        \else
            \unskip
            \ALG@printindent@tempcnta=1
            \loop
                \algrule[\csname ALG@ind@\the%
                    \ALG@printindent@tempcnta\endcsname]%
                \advance \ALG@printindent@tempcnta 1
            \ifnum \ALG@printindent@tempcnta<\numexpr\theALG@nested+1\relax%
            \repeat
        \fi
    \fi
    }%
\patchcmd{\ALG@doentity}{\noindent\hskip\ALG@tlm}%
        {\ALG@printindent}{}{\errmessage{failed to patch}}
\makeatother

\newcommand{\keyword}[1]{\textbf{#1}}

\newcounter{main}
\newtheorem{prop}[main]{Proposition}
\newtheorem{comp}[main]{Computation}
\newtheorem{cor}[main]{Corollary}

\newtheorem{lem}[main]{Lemma}
\newtheorem{fact}[main]{Fact}
\theoremstyle{definition}
\newtheorem{dfn}[main]{Definition}
\newtheorem{exa}[main]{Example}
\newtheorem*{spin}{SPIN Axiom \cite{ck09}}
\theoremstyle{remark}

\title{A Kochen-Specker system has at least 22 vectors\\
        {\small (extended abstract)%
            \footnote{This is a condensed version of the article
                        to be published in Ohmsha-Springer's
                        New Generation Computing.
                    One can find a preprint of the full version at
                    \url{http://westerbaan.name/\~bas/math/ks.pdf}}}}
\def\titlerunning{A Kochen-Specker system has at least 22 vectors
            (extended abstract)}
\def\authorrunning{Uijlen \& Westerbaan}

\author{Sander Uijlen
    \institute{Radboud Universiteit}
    \email{suijlen@cs.ru.nl}
\and
    Bas Westerbaan
    \institute{Radboud Universiteit}
    \email{bwesterb@cs.ru.nl}}

\begin{document}

\maketitle

\begin{abstract}
    At the heart of the Conway-Kochen
    Free Will Theorem and Kochen and Specker's
        argument against non-contextual hidden variable theories
    is the existence of a Kochen-Specker (KS) system:
    a set of points on the sphere
    that has no~$\{0,1\}$-coloring such that
    at most one of two orthogonal points are colored~$1$
    and of three pairwise orthogonal points exactly one
    is colored~$1$.
    In public lectures, Conway encouraged the search for small
    KS systems.
    At the time of writing, the smallest known
    KS system has 31 vectors.  

    Arends, Ouaknine and Wampler have shown that a KS system has at least
    18 vectors, by reducing the problem to the existence of graphs
    with a topological embeddability and non-colorability property.
    The bottleneck in their search
    proved to be the sheer number of graphs on more than~$17$
    vertices and deciding embeddability.

    Continuing their effort, we prove a restriction on the class of graphs
    we need to consider and develop a more practical decision procedure for
    embeddability to improve the lower bound to 22.
\end{abstract}

\section{Introduction}

\subsection{The experiment}

Consider the following experiment.  Shoot a deuterium
atom (or another neutral spin 1 particle)
through a certain fixed inhomogeneous magnetic field,
such as that in the Stern-Gerlach experiment.
The particle will then move undisturbed or deviate.
What we have done is measure the spin component%
\footnote{%
As we are only interested in whether the particle deviates or
    not, we actually only consider the square of the spin component.}
 of the particle
along a certain direction.  This direction depends on the specifics of the
field and the movement of the particle.

Quantum Mechanics only predicts the probability, given the direction,
whether the particle will deviate.
Its probabilistic prediction has been thoroughly tested.
One wonders: is there a
\emph{deterministic} theory predicting the
outcome of this experiment?

Kochen and Specker have shown that such a non-contextual deterministic
theory must be odd: it cannot satisfy the plausible SPIN
axiom, that is:
\begin{spin}
    Given three pairwise orthogonal directions.
    In exactly one of the directions, the particle will not deviate.
\end{spin}
Their argument is based on the existence of a Kochen-Specker system.
\begin{dfn}
    A \keyword{Kochen-Specker (KS) system} is
    a finite set of points on the sphere\footnote{
            We define KS systems to be three dimensional,
            as in the original proof of Kochen and Specker.
            Later, higher dimensional systems have been studied.
            See, for instance~\cite[p.~201]{qtcm}.
        }
    for which each pair is not antipodal and
    there is no~\keyword{010-coloring}.
    A $010$-coloring is a~$\{0,1\}$-coloring of the points such that%
        \footnote{
                In other papers, like \cite{aow11},
                the~$0$ and~$1$ are swapped; they consider 101-colorings.
                These colorings are of course equivalent and the
                difference arises from considering either squared
                spin measurements $S^2$, or $1-S^2$.
               	}
    \begin{enumerate}
        \item
            no pair of orthogonal points are both colored~$1$ and
        \item
            of three pairwise orthogonal points exactly one is colored~$1$;
            or alternatively: they are colored~$0$, $1$ and~$0$ in some order.
    \end{enumerate}
\end{dfn}
A point on the sphere obviously corresponds to a direction in space.
Because of this, the term point, vector and direction
can be used interchangeably. Antipodal points correspond to opposite
vectors and these span the same direction in space.

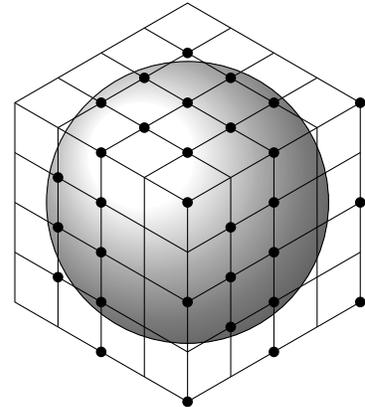
\piccaption{Conway's 31 vector Kochen-Specker system\label{fig:conway31}}
\parpic[r]{%
\begin{tikzpicture}[scale=0.75,
        style={inner sep=0pt,%
              outer sep=2pt,%
              mark coordinate/.style={%
                inner sep=0pt,outer sep=0pt,minimum size=4pt,
                fill=black,circle}%
            }]
\def\R{2.5} 
\def\St{1.41421356237}
\def\ST{1.73205080757}
\filldraw[ball color=white] (0,0) circle (\R);
\coordinate[mark coordinate] (N) at (0,0);
\coordinate[mark coordinate] (N) at (0,-\R/\St);
\coordinate[mark coordinate] (N) at (0,-\R*\St);
\coordinate[mark coordinate] (N) at (0,\R*\St/4);
\coordinate[mark coordinate] (N) at (0,\R*\St/2);
\coordinate[mark coordinate] (N) at (0,3*\R*\St/4);
\coordinate[mark coordinate] (N) at (-\R*\St*\ST/8,3*\R*\St/8);
\coordinate[mark coordinate] (N) at (-\R*\St*\ST/8,5*\R*\St/8);
\coordinate[mark coordinate] (N) at (-\R*\St*\ST/4,0);
\coordinate[mark coordinate] (N) at (-\R*\St*\ST/4,-\R*\St/4);
\coordinate[mark coordinate] (N) at (-\R*\St*\ST/4,-\R*\St/2);
\coordinate[mark coordinate] (N) at (-\R*\St*\ST/4,-3*\R*\St/4);
\coordinate[mark coordinate] (N) at (-\R*\St*\ST/4,\R*\St/4);
\coordinate[mark coordinate] (N) at (-\R*\St*\ST/4,\R*\St/2);
\coordinate[mark coordinate] (N) at (-3*\R*\St*\ST/8,\R*\St/8);
\coordinate[mark coordinate] (N) at (-3*\R*\St*\ST/8,-\R*\St/8);
\coordinate[mark coordinate] (N) at (-3*\R*\St*\ST/8,-3*\R*\St/8);
\coordinate[mark coordinate] (N) at (\R*\St*\ST/8,3*\R*\St/8);
\coordinate[mark coordinate] (N) at (\R*\St*\ST/8,5*\R*\St/8);
\coordinate[mark coordinate] (N) at (\R*\St*\ST/8,-\R*\St/8);
\coordinate[mark coordinate] (N) at (\R*\St*\ST/8,-3*\R*\St/8);
\coordinate[mark coordinate] (N) at (\R*\St*\ST/8,-5*\R*\St/8);
\coordinate[mark coordinate] (N) at (\R*\St*\ST/4,0);
\coordinate[mark coordinate] (N) at (\R*\St*\ST/4,\R*\St/4);
\coordinate[mark coordinate] (N) at (\R*\St*\ST/4,\R*\St/2);
\coordinate[mark coordinate] (N) at (\R*\St*\ST/4,-\R*\St/4);
\coordinate[mark coordinate] (N) at (\R*\St*\ST/4,-\R*\St/2);
\coordinate[mark coordinate] (N) at (\R*\St*\ST/4,-3*\R*\St/4);
\coordinate[mark coordinate] (N) at (\R*\St*\ST/2,-\R/\St);
\coordinate[mark coordinate] (N) at (\R*\St*\ST/2,0);
\coordinate[mark coordinate] (N) at (\R*\St*\ST/2,\R/\St);
\draw (0,0) -- (0,-\R*\St) -- (-\R*\St*\ST/2,-\R/\St) --  (-\R*\St*\ST/2,\R/\St)
    --  (0,\R*\St) -- (\R*\St*\ST/2,\R/\St) 
    -- (\R*\St*\ST/2,-\R/\St) -- (0,- \R*\St);
\draw  (-\R*\St*\ST/2,\R/\St) -- (0,0) -- (\R*\St*\ST/2,\R/\St);
\draw (-\R*\St*\ST/4, -3*\R*\St/4) -- (-\R*\St*\ST/4,\R*\St/4 )
    -- (\R*\St*\ST/4, 3*\R*\St/4) ;
\draw (\R*\St*\ST/4, -3*\R*\St/4) -- (\R*\St*\ST/4,\R*\St/4 )
    -- (-\R*\St*\ST/4, 3*\R*\St/4) ;
\draw  (-\R*\St*\ST/2,{\R/\St - \R*\St/4}) -- (0,- \R*\St/4)
    -- (\R*\St*\ST/2,{\R/\St - \R*\St/4});
\draw  (-\R*\St*\ST/2,{\R/\St - \R*\St/2}) -- (0,- \R*\St/2)
    -- (\R*\St*\ST/2,{\R/\St - \R*\St/2});
\draw  (-\R*\St*\ST/2,{\R/\St - 3*\R*\St/4}) -- (0,- 3*\R*\St/4)
    -- (\R*\St*\ST/2,{\R/\St - 3*\R*\St/4});
\draw(-\R*\St*\ST/8,-0.25*7*\R/\St) -- (-\R*\St*\ST/8,{-0.25*7*\R/\St + \R*\St})
    -- (3*\R*\St*\ST/8,5*0.25*\R/\St);
\draw(\R*\St*\ST/8,-0.25*7*\R/\St) -- (\R*\St*\ST/8,{-0.25*7*\R/\St + \R*\St})
    -- (-3*\R*\St*\ST/8,5*0.25*\R/\St);
\draw (-3*\R*\St*\ST/8,-5*0.25*\R/\St)
    -- (-3*\R*\St*\ST/8,{-5*0.25*\R/\St + \R*\St})
    -- (\R*\St*\ST/8,0.25*7*\R/\St);
\draw (3*\R*\St*\ST/8,-5*0.25*\R/\St)
    -- (3*\R*\St*\ST/8,{-5*0.25*\R/\St + \R*\St})
    -- (-\R*\St*\ST/8,0.25*7*\R/\St);
\undef\R
\undef\ST
\undef\St
\end{tikzpicture}}

Suppose there is a KS system and a non-contextual deterministic theory satisfying
the SPIN Axiom.
Then we color a point of this system~$0$,
whenever this theory predicts that the particle will deviate
if the spin is measured in the direction corresponding to that
point, and~$1$ otherwise.
Given two orthogonal points of the system,
we can find a third point orthogonal to both of them.
The SPIN axiom implies exactly one of them is colored~$1$, so they
cannot both be colored~$1$.
Similarly, given three pairwise orthogonal vectors in the system,
the SPIN axiom implies exactly one of them is colored~$1$.
Hence there would be a 010-coloring of the KS system, quod non.
Therefore a deterministic non-contextual theory cannot satisfy the
SPIN Axiom.

The KS system proposed by Kochen and Specker contained 117 points\cite{ks}.
Penrose and Peres\cite{peres} independently found a smaller system of 33 points.
The current record is the 31 point system of Conway\cite[p.~197]{qtcm}.
As pointed out by \cite{c00,aow11}, finding small KS systems
is of both theoretical and practical interest.
In public lectures, Conway himself, stressed the search for small KS
systems.\cite{OC}

\subsection{Overview}
In \cite{aow11} Arends, Ouaknine and Wampler (AOW) give a computer aided proof
that a KS system must have at least 18 vectors.  We improve their lower bound
and show that a KS system must have at least 22 vectors.

First, in Subsection~\ref{sec:ksgraphs},
we repeat a part of AOW's work, in particular the reduction of
KS systems to graphs.
The bottleneck of their search was the sheer number of graphs
and the deciding whether such graphs are embeddable.
In Section~\ref{sec:ilb},
we improve upon their reduction,
to cut down the number of graphs to consider drastically,
and state the results of our main computation.
Finally, in Section~\ref{sec:emb},
we describe our practical embeddability test.

The software and results of the various computations performed for
this paper, can be found here\cite{GH}.

\subsection{Kochen-Specker graphs}
\label{sec:ksgraphs}
We follow \cite{aow11} and reduce the search for Kochen-Specker systems
to the search of a certain class of graphs.
First note that in a Kochen-Specker system we may replace a point with its
antipodal point.  They are both orthogonal to the same points and hence
the non-010-colorability is preserved.
Therefore, we may assume antipodal points are identified on the sphere.
That is: a Kochen-Specker system is a finite subset of the projective plane
that is not 010-colorable.

\begin{dfn}
Given a finite subset~$S$ of the projective plane
(or equivalently, a finite subset of the northern
hemisphere without equator\footnote{%
    A subset of the projective plane can be identified with
    a subset of the closed northern hemisphere.
    For a finite subset we can always rotate in such a way
    that no points lie on the equator.}).
Define its \keyword{orthogonality graph}~$G(S)$ as follows.
The vertices are the points of~$S$.
Two vertices are joined by an edge, if their corresponding points
are orthogonal.
\end{dfn}
\begin{dfn}
A graph~$G$ is called~\keyword{embeddable},
if it occurs as a subgraph of an orthogonality graph.
That is: if there is a finite subset~$S$ of the projective plane,
such that~$G \leq G(S)$.
\end{dfn}
\begin{dfn}
A graph is called~\keyword{010-colorable}
if there is a~$\{0,1\}$-coloring of the vertices,
such that
\begin{enumerate}
\item
for each triangle there is exactly one vertex that is colored~$1$ and
\item
adjacent vertices are not both colored~$1$.
\end{enumerate}
\end{dfn}

\begin{dfn}
A \keyword{Kochen-Specker graph}
is a embeddable graph that is not 010-colorable.
\end{dfn}
It is an easy, but important, consequence of the definitions that:
\begin{fact}
    A finite subset~$S$ of the projective plane
    is a Kochen-Specker system,
    if and only if its orthogonality graph~$G(S)$
    is Kochen-Specker.
\end{fact}

To prove there is no Kochen-Specker system on~$17$ points,
it would be sufficient to enumerate all graphs on~$17$ vertices
and check these are not 010-colorable or not embeddable.
However, this is infeasible as there are
already~${\sim}10^{26}$ non-isomorphic
graphs on~$17$ points.\cite{oeisA000088}
Luckily, we can restrict ourselves to certain classes of graphs.
\begin{prop}[\cite{aow11}]
    An embeddable graph is squarefree.
    That is: it does not contain the square as a subgraph.%
    \footnote{Some authors call a graph squarefree if it does not
        contain the square as induced subgraph.
        For them the complete graph on four vertices is squarefree.
        We follow~Weisstein\cite{sf-weisstein} and~Sloane\cite{sf-sloane} and
        call a graph squarefree if it does not
        contain the square as subgraph.
        For us the complete graph on four vertices is not squarefree.}
\end{prop}

\parpic[r]{%
\begin{tikzpicture}[scale=1.0,thick,
    partial ellipse/.style args={#1:#2:#3}{
            insert path={+ (#1:#3) arc (#1:#2:#3)}
        }]
    \draw (0,0) [partial ellipse=360:180:1 and 0.3] ;
    \draw[gray,dashed] (0,0) [partial ellipse=0:180:1 and 0.3] ;
    \draw (0,0) [rotate=315,partial ellipse=360:180:1 and 0.3] ;
    \draw[gray,dashed] (0,0) [rotate=315,partial ellipse=0:180:1 and 0.3] ;
    \draw (0,0) circle (1) ;
    \node[circle,inner sep=1pt,fill] at (90:1) {};
    \node[above] at (90:1) {$v$};
    \node[circle,inner sep=1pt,fill] at (45:1) {};
    \node[right] at (45:1) {$w$};
\end{tikzpicture}}

\begin{proof}
    Given two non antipodal points~$v\neq w$.
    See the figure on the right.
    Consider the points orthogonal to~$v$.
    This is a great circle.
    The points orthogonal to~$w$ is a different great circle.
    They intersect in precisely two antipodal points.
    Hence, if~$c$ and~$d$ are both orthogonal to~$v$ and~$w$,
    then~$c$ and~$d$ are equivalent.
    Therefore, an embeddable graph cannot contain a square.
\end{proof}

The squarefreeness is a considerable restriction.  There are
only~${\sim}10^{10}$ non-isomorphic squarefree graphs on~$17$
vertices.\cite{sf-sloane}
Next, we show we can restrict ourselves to connected graphs.
\begin{prop}[\cite{aow11}]\label{prop:ks-conn}
    A minimal Kochen-Specker graph is connected.
\end{prop}
\begin{proof}
    Suppose~$G$ is a non-connected Kochen-Specker graph.
    Then one of its components is not 010-colorable.
    As a subgraph of an embeddable graph, is embeddable,
    this component is embeddable as well.
    Hence it is a smaller connected Kochen-Specker graph.
\end{proof}
The gain, however, is small.
There are only~${\sim}10^9$ non-isomorphic squarefree graphs on~$17$
vertices that are not connected.
In our computations, checking for connectedness
required more time than would be gained by reducing the number of graphs.

We have verified the main result of \cite{aow11}:
\begin{comp}
There is a unique non-010-colorable squarefree connected graph on~$17$ or less
vertices:
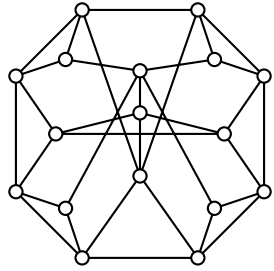
\begin{center}
    \begin{tikzpicture}[thick,scale=1.4,
        p/.style={circle, draw, fill=white,
                        inner sep=0pt, minimum width=5pt}]
        \draw
            \foreach \x in {0,90,180,270} {
                (\x+45:1) node[p]{} -- (\x+65:1.3) node[p]{}
                (\x+45:1)  -- (\x+25:1.3) node[p]{}
                (\x+65:1.3)  -- (\x+25:1.3)
            }
            \foreach \x in {0,180} {
                (\x:0.8) node[p]{} -- (\x+25:1.3)
                (\x:0.8) -- (\x-25:1.3)
                (\x+25:1.3) -- (\x-25:1.3)
                (\x+65:1.3) -- (\x+115:1.3)
            }
            (0:0.8) -- (180:0.8)
            (90:0.2) node[p]{} -- (0:0.8)
            (90:0.2) -- (180:0.8)

            (90:0.2) -- (90:0.6) node[p]{}
            (90:0.6) -- (270+45:1)
            (90:0.6) -- (270-45:1)
            (90:0.6) -- (90+45:1)
            (90:0.6) -- (90-45:1)

            (90:0.2) -- (270:0.4) node[p]{}
            (270:0.4) -- (270+25:1.3)
            (270:0.4) -- (270-25:1.3)
            (270:0.4) -- (90+25:1.3)
            (270:0.4) -- (90-25:1.3)
        ;
    \end{tikzpicture}
\end{center}
It is not embeddable, as the graph in Figure~\ref{fig:unemb-10-2}
is an unembeddable subgraph.  For our proof,
see Proposition~\ref{prop:unemb-10-2}.
Hence a Kochen-Specker
system has at least 18 points.
\end{comp}

\section{An improved lower bound}
\label{sec:ilb}
Continuing the effort of Arends, Ouaknine and Wambler,
we consider another restriction.
\begin{prop}
    A minimal Kochen-Specker graph has minimal vertex-order three.
    That is: every vertex is adjacent to at least three other vertices.
\end{prop}
\begin{proof}
    Given a minimal Kochen-Specker graph~$G$.
    Suppose~$v$ is a vertex with order less than or equal~$2$.
    Let~$G'$ be~$G$ with~$v$ removed.
    Clearly~$G'$ is embeddable.
    Suppose~$G'$ is 010-colorable.
    Then we can extend the coloring to a coloring of~$G$ as follows.
    If~$v$ is adjacent to only one or no vertex,
    then we can color~$v$ with~$0$.
    Suppose~$v$ is adjacent to two vertices, say~$w$ and~$w'$.
    If one of~$w$ or~$w'$ is colored~$1$, we can color~$v$ with~$0$.
    If both~$w$ and~$w'$ are colored~$0$, we can color~$v$ with~$1$.
    This would imply~$G$ is 010-colorable, quod non.
    Therefore~$G'$ is a smaller
    Kochen-Specker graph, which contradicts the minimality of~$G$.
\end{proof}
There are only~${\sim}10^7$
squarefree non-isomorphic graphs on 17 vertices with minimal vertex order 3.
Even though Arends, Ouaknine and Wampler
note this restriction once,
surprisingly, they did not restrict their graph enumeration
to graphs with minimal vertex order 3.

We continue with a strengthening of Proposition~\ref{prop:ks-conn}.
\begin{prop}\label{prop:ks-biconn}
A minimal Kochen-Specker graph is edge-biconnected.
That is: removing any single edge leaves the graph connected.
\end{prop}
We need some preparation, before we can prove this Proposition.
\begin{dfn}
Given a graph~$G$ and a vertex~$v$ of~$G$.
We say, \keyword{$v$ has fixed color~$c$ (in~$G$)},
if~$G$ is~$010$-colorable
and for every~$010$-coloring of~$G$,
the vertex~$v$ is assigned color~$c$.
\end{dfn}
We are interested in these graphs because of the following observation.
\begin{lem}\label{lem:biconn}
If there is an embeddable graph~$G$ on~$n$ vertices with a vertex
with fixed color~$1$,
then there is a Kochen-Specker graph on~$2n$ vertices.
\end{lem}
\begin{proof}
Let~$G$ be a graph and~$v$ a vertex of~$G$ with fixed color~$1$.
Consider two copies of the graph~$G$.
Connect the two instances of~$v$ with an edge.
Call this graph~$G'$.
Clearly, $G'$ is not~$010$-colorable.

We need to show~$G'$ is embeddable.
Given an embedding~$S$ of~$G$.
We may assume that the point in~$S$ corresponding to~$v$
is the north pole.
Furthermore, we may assume that there is no point on the $x$-axis,
by rotating points along the north pole.
Let~$S'$ be~$S$ rotated~$90$ degrees along the~$y$-axis.
Some points of~$S$ and~$S'$ might overlap.
That is: there might be a point~$s$ in~$S$ and~$s'$ in~$S'$
that are equal or antipodal.
Observe that if no points of~$S'$ and~$S$ overlap,
then~$S \cup S'$ is an embedding of~$G'$.

Suppose there are points in~$S'$ and~$S$
that overlap.
Note that the north pole (and south pole) is not in~$S'$. 
Let~$S''$ be~$S'$ rotated along the north pole at some angle~$\alpha$.
There are finitely many angles such that there are overlapping points.
Thus there is an angle such that~$S \cup S''$ is an embedding of~$G'$.
\end{proof}
Unfortunately, these graphs are not small.

\begin{comp}\label{comp:bic1}
    There are no embeddable graphs with fixed color~$1$
    on less than 17 vertices.%
        \footnote{Source code at \texttt{code/comp5.py} of \cite{GH}.}
\end{comp}

We are ready to prove that a minimal Kochen-Specker graph
is edge-biconnected.
\begin{proof}[Proof of Proposition~\ref{prop:ks-biconn}]
Given a minimal Kochen-Specker graph~$G$.
\parpic(0cm,1cm)[r]{%
    \begin{tikzpicture}[thick,scale=1.4,
        p/.style={circle, draw, fill=white,
                        inner sep=0pt, minimum width=5pt}]
            \draw [rotate=90] (0,-0.8) ellipse (0.6 and 0.5) ;
            \draw [rotate=90] (0,0.8) circle (0.6 and 0.5) ;
            \draw 
                  (-0.7, 0.1) node[left] {$a$}
                  (0.7, 0.1) node[right] {$b$}
                  (-0.7, 0.1) node[p] {} -- (0.7, 0.1) node[p] {}
                  (-0.8, -0.6) node[below] {$A$}
                  (0.8, -0.6) node[below] {$B$}
                ;
        \end{tikzpicture}}
Recall it must be connected.  Suppose it is not edge-biconnected.
Then there must be an edge~$(a,b)$ in~$G$,
which removal disconnects~$G$.
Thus~$G$ decomposes into two connected graphs~$A$ and~$B$
such that~$a\in A$, $b \in B$ and~$(a,b)$ is the only edge
between~$A$ and~$B$.
Clearly~$A$ and~$B$ are embeddable.

Note that~$A$ must be~$010$-colorable,
for if it were not~$010$-colorable,
then~$A$ is a Kochen-Specker graph,
in contradiction with~$G$'s minimality.
Similarly~$B$ is~$010$-colorable.
Suppose there is a~$010$-coloring of~$A$ in which~$a$ is colored~$0$.
Then we can extend this coloring with any $010$-coloring of~$B$
to a~$010$-coloring of~$G$, which is absurd.
Thus~$a$ must have fixed color~$1$ in~$A$.
Similarly~$b$ must have fixed color~$1$ in~$B$.
Thus by Computation~\ref{comp:bic1},
we have~$\#A \geq 17$ and~$\#B \geq 17$.
Consequently~$\#G \geq 34$.
Contradiction with~$G$'s minimality.
\end{proof}

We can go one step further.

\begin{prop}\label{prop:ks-triconn}
    A minimal Kochen-Specker graph is edge-triconnected.
    That is: removing any two edges keeps the graph connected.
\end{prop}

Again, we need some preparation for the proof.  First, we generalize
the notion of fixed color.

\begin{dfn}
   Given a graph~$G$ together with selected vertices~$v_1, \ldots, v_n \in G$.
   Let~$C(G)$ denote the set of $010$-colorings of~$G$.
   The \keyword{type~$t$ of~$(v_1,\ldots,v_n)$ (in~$G$)}
   is the set of all possible ways $\{v_1, \ldots, v_n\}$ can be colored.
   That is: $t = \{ (c(v_1), \ldots, c(v_n));
                \ c \in C(G) \}.$ 
    A type of~$n$ vertices is
    called an~\keyword{$n$-type}.
\end{dfn}

\begin{exa}
\begin{itemize}
\item The triangle has~$3$-type~$\{(1,0,0), (0,1,0), (0,0,1)\}$.
\item Every vertex in a Kochen-Specker graph has type~$\emptyset$.
\item
    A vertex~$v$ has the~$1$-type~$\{(1)\}$ in~$G$ if and only if
        it has fixed color~$1$ in~$G$.
\end{itemize}
\end{exa}

Just as verteces with fixed color are rare,
we are interested in types,
because most types do not occur in small graphs.

\begin{comp}\label{comp:types}
    We have enumerated all embeddable graphs of less than~$17$ vertices
    and determined a lower bound at which a particular $1$- or $2$-type occurs,
    omitting the trivial types~$\{(0),(1)\}$
    and~$\{(0,0),(0,1),(1,0),(1,1)\}$.%
        \footnote{Source code at \texttt{code/comp5.py} of \cite{GH}.}
    \begin{center}
    \begin{tabular}{ll}
        $1/2$-type & $\#G$ \\ \hline
        $\{(0,0), (1,0), (0,1)\}$ &
            non-trivially $\geq 10$\\ 
        $\{(0,0), (1,0), (1,1)\}$ & $\geq 10$ \\
        $\{(0,0), (0,1), (1,1)\}$ & $\geq 10$ \\
        $\{(0,0), (0,1))\}$ & $\geq 15$ \\
        $\{(0,0), (1,0))\}$ & $\geq 15$ \\
        $\{(0)\}$ & $\geq 15$ \\
        $\{(0,1),(1,0)\}$ & $\geq 16$ \\
        other & $\geq 17$
    \end{tabular}
    \end{center}
    The type~$\{(0,0),(1,0),(0,1)\}$ occurs in the embeddable two-vertex
    graph
    \begin{tikzpicture}[thick,scale=1.4,
        p/.style={circle, draw, fill=white,
                        inner sep=0pt, minimum width=5pt}]
            \draw
            (0,0) node[p]{} -- (0.3,0) node[p]{} ;
    \end{tikzpicture}.
    Because the two vertices are adjacent, this occurance of
    the type is called trivial.
\end{comp}

\begin{proof}[Proof of Proposition~\ref{prop:ks-triconn}]

    Given a minimal Kochen-Specker graph~$G$.

\parpic(0cm,1cm)[r]{%
\begin{tikzpicture}[thick,scale=1.4,
p/.style={circle, draw, fill=white,
                inner sep=0pt, minimum width=5pt}]
    \draw [rotate=90] (0,-0.8) ellipse (0.8 and 0.5) ;
    \draw [rotate=90] (0,0.8) circle (0.8 and 0.5) ;
    \draw 
          (-0.7, 0.3) node[left] {$a_1$}
          (-0.7, -0.3) node[left] {$a_2$}
          (0.7, 0.3) node[right] {$b_1$}
          (0.7, -0.3) node[right] {$b_2$}
          (-0.7, -0.3) node[p] {} -- (0.7, -0.3) node[p] {}
          (-0.7, 0.3) node[p] {} -- (0.7, 0.3) node[p] {}
          (-0.8, -0.8) node[below] {$A$}
          (0.8, -0.8) node[below] {$B$}
        ;
\end{tikzpicture}}

    Suppose it is not edge-triconnected.
    Then it splits into two graphs~$A$ and~$B$
    together with verteces~$a_1, a_2 \in A$
    and~$b_1,b_2 \in B$
    such that~$(a_1, b_1)$ and~$(a_2,b_2)$
    are the only edges between~$A$ and~$B$.
    Note that~$A$ and~$B$ must be $010$-colorable,
    for otherwise~$G$ would not be a minimal Kochen-Specker graph.

    \begin{enumerate}
    \item Suppose~$a_1 = a_2$ and~$b_1 = b_2$.
                Then~$G$ is not edge-biconnected.  Contradiction with
                    Proposition~\ref{prop:ks-biconn}.
    \item Suppose~$a_1 \neq a_2$ and~$b_1 = b_2$.
        \label{tricase1}
        Suppose~$b_1=b_2$ does not have a fixed color in~$B$.
        Then any coloring of~$A$ can be extended with some coloring
        in~$B$ to a coloring of~$G$.  Contradiction.
        Apparently~$b_1=b_2$ has a fixed color in~$B$.
        \begin{enumerate}
            \item
            Suppose~$b_1=b_2$ has fixed color~$1$ in~$B$.
            Note~$\#B \geq 17$ by Computation~\ref{comp:types}.

            Suppose there is a coloring of~$A$
            in which both~$a_1$ and~$a_2$ have color~$0$.
            Then, regardless whether~$a_1$ and~$a_2$ are adjacent or not,
            this coloring can be extended with a coloring of~$B$
            (in which~$b_1=b_2$ must be colored~$1$)
            to a coloring~$G$.  Contradiction.

            Thus the type of~$(a_1,a_2)$ in~$A$ cannot contain~$(0,0)$.
            Thus, by Computation~\ref{comp:types}, $\#A \geq 17$.
            Consequently~$\#G \geq 34$. Contradiction with minimality.

            \item
            Apparently~$b_1=b_2$ has fixed color~$0$ in~$B$.
            Hence, by Computation~\ref{comp:types}, $\#B \geq 15$.

            Suppose~$a_1$ is not adjacent to~$a_2$.
            Then any coloring of~$A$ can be extended with a coloring
            of~$B$ to a coloring of~$G$. Contradiction.
            
            Apparently~$a_1$ is adjacent to~$a_2$.

            The type of~$(a_1,a_2)$ in~$A$ cannot contain~$(1,0)$
            or~$(0,1)$ for otherwise~$G$ can be colored.
            It also cannot contain~$(1,1)$ as~$a_1$ and~$a_2$
            are adjacent.
            Thus both~$a_1$ and~$a_2$ have fixed color~$0$ in~$A$.
            Hence~$\#A\geq 17$ by Computation~\ref{comp:types}.
            Consequently~$\#G \geq 32$. Contradiction
            with minimality.
        \end{enumerate}
    \item Suppose~$a_1 = a_2$ and~$b_1 \neq b_2$.
            This leads to a contradication in the same way
            as in case~\ref{tricase1}.
    \item Apparently~$a_1 \neq a_2$ and~$b_1 \neq b_2$.
        The type of~$(a_1,a_2)$ in~$A$ cannot contain~$(0,0)$,
        for otherwise~$G$ is colorable.
        Similarly, the type of~$(b_1,b_2)$ in~$B$ cannot contain~$(0,0)$.
        Thus both~$\#A \geq 17$ and~$\#B \geq 17$.
        Hence~$\#G \geq 34$. Contradiction with minimality. \qedhere
    \end{enumerate}
\end{proof}

Although these restrictions are theoretically pleasing,
they seem to be of little use as a practical restriction.
Concerning excluding unconnected graphs:
\begin{comp}
    There are five non-isomorphic minimal
    squarefree connected graphs
    with minimal vertex order 3 and they have 10 vertices.
\end{comp}
\begin{cor}
    Any unconnected
    squarefree graph with minimal vertex order 3
    has at least 20 vertices, for it has two connected components,
    each with at least 10 vertices.
    With 20 vertices, there are exactly 25 of these.
\end{cor}
This justifies, at this stage, not checking for connectedness.
Similarly, we believe there are very few connected but not
edge-biconnected graphs.

Now we can state our main computation.
\begin{comp}
    Let~$C_n$ denote the number of non-010 colorable squarefree
    graphs with minimal vertex order 3 on~$n$ nodes.  Then:%
        \footnote{Source code at \texttt{code/comp6} of \cite{GH}.}

    \begin{center}
    \begin{tabular}{l|llllll}
        $n$ & $\leq 16$
            & $17$
            & $18$
            & $19$
            & $20$
            & $21$ \\
        \hline
        $C_n$ & $0$
            & $1$
            & $2$
            & $19$
            & $441$
            & $11876$
    \end{tabular}
    \end{center}

    All these 12339 graphs are not embeddable.
    See Computation~\ref{comp:unemb20}.
\end{comp}
The computation was distributed on approximately 300 CPU cores
and took roughly three months.
It was executed as follows.
We enumerated all squarefree graphs with minimal vertex
order 3 on less than or equal~$21$ vertices,
using the~\texttt{geng} util of the nauty software package,
which uses the isomorphism-free exhaustive generation
method of McKay\cite{geng}.
The output of~\texttt{geng}, we passed through
a custom heuristic backtracker written in~C++
to decide 010-colorability of these graphs.

\section{Embeddability}\label{sec:emb}
Our computation has yielded over nine-thousand non-010-colorable graphs.
If we show one of them is embeddable, we have found a new KS system.
If we demonstrate all of them are not embeddable, we have
proven a lower bound on the size of a minimal KS system.

In~\cite{aow11}, Arends, Wampler and Ouaknine discuss several
computer-aided methods
to test embeddability of a graph.  None of these methods could decide
for all graphs considered, whether they were embeddable or not.

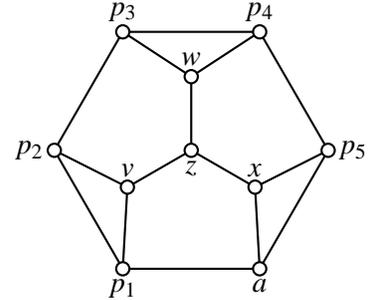
\piccaption{One of the two minimal non-embeddable graphs \label{fig:unemb-10-2}}
\parpic[r]{%
\begin{tikzpicture}[thick,scale=1.4,%
        p/.style={circle, draw, fill=white,%
                        inner sep=0pt, minimum width=5pt}]
            \draw
                \foreach \x in {90,210,330} {
                    (\x:0.7) -- (0:0)
                    (\x-30:1.3) -- (\x:0.7)
                    (\x+30:1.3) -- (\x:0.7)
                    (\x-30:1.3) -- (\x+30:1.3) 
                    (\x+30:1.3) -- (\x-30+120:1.3) 
                }
                (90:0.7) node[p]{}
                (210:0.7) node[p]{}
                (330:0.7) node[p]{}
                (0:0) node[p]{}
                (90+30:1.3) node[p]{}
                (90-30:1.3) node[p]{}
                (210+30:1.3) node[p]{}
                (210-30:1.3) node[p]{}
                (330+30:1.3) node[p]{}
                (330-30:1.3) node[p]{}
                (90:0.7) node[above]{$w$}
                (210:0.7) node[above]{$v$}
                (330:0.7) node[above]{$x$}
                (0:0) node[below]{$z$}
                (90+30:1.3) node[above]{$p_3$}
                (90-30:1.3) node[above]{$p_4$}
                (210+30:1.3) node[below]{$p_1$}
                (210-30:1.3) node[left]{$p_2$}
                (330+30:1.3) node[right]{$p_5$}
                (330-30:1.3) node[below]{$a$} ;
\end{tikzpicture}}

We propose a new method,
which for all graphs we considered,
could decide
whether they were embeddable or not.
First we give a pen-and-paper example.
\begin{prop}\label{prop:unemb-10-2}
The graph in Figure~\ref{fig:unemb-10-2}
is not embeddable.
\end{prop}
\begin{proof}
Suppose it is embeddable.
Consider~$p_1$.
It is orthogonal to both~$a$ and~$v$.
Since $a$ and~$v$ are not collinear,
$p_1$ must be collinear to~$v \times a$,
the cross-product of~$v$ and~$a$.
Similarly, $p_2$ is collinear to~$v \times p_1 = v \times (v \times a)$.
Continuing in this fashion,
we see that
\begin{equation}\label{eq:ue1}
    a \text{ is collinear to }
    x \times (x \times( w \times (w\times (v \times (v \times a))))).
\end{equation}
Now, we may assume that~$z=(0,0,1)$ and $x=(1,0,0)$.
Thus: $v=(v_1,v_2,0)$;
$w = (w_1,w_2,0)$
and $a = (0, a_2,a_3)$ for some~$-1 \leq v_1,v_2,w_1,w_2,a_2,a_3 \leq 1$,
with~$v_1^2+v_2^2 = 1$; $w_1^2+w_2^2=1$ and~$a_2^2 + a_3^2=1$.
Now, \eqref{eq:ue1} becomes:
\begin{equation*}
\begin{pmatrix}
0\\
a_2\\
a_3\\
\end{pmatrix}
\text{ is collinear to }
\begin{pmatrix}
0\\
-a_2 v_1w_2 (v_1w_1 + v_2w_2) \\
-a_3 (v_1^2 w_1^2 + v_1^2 w_2^2 + v_2^2w_1^2 + v_2^2w_2^2)\\
\end{pmatrix}.
\end{equation*}
Consequently
\begin{align*}
    v_1w_2 (v_1w_1 + v_2w_2)
& = v_1^2 w_1^2 + v_1^2 w_2^2 + v_2^2w_1^2 + v_2^2w_2^2 \\
& = (v_1^2 + v_2^2)w_1^2 + (v_1^2 + v_2^2)w_2^2 \\
& = w_1^2 + w_2^2 \\
&= 1.
\end{align*}
Since~$v$ and~$w$ are not collinear,
we have by Cauchy-Schwarz $|\left<v,w\right>| < 1$.
Now we find the contradiction:
\begin{equation*}
1 > |v_1w_2\left<v,w\right>| = |v_1w_2(v_1w_1+v_2w_2)| = 1. \qedhere
\end{equation*}
\end{proof}


In the previous proof, we fixed, without loss of generality, the position
of a few vertices.  Then we derived cross-product expressions for the
remaining vertices.  Finally, we find an equation relating some of
the cross-product expressions and show it is unsatisfiable.
We can automate this reasoning as follows.

\begin{algorithmic}[5]
    \While{there are unassigned vertices}
        \State pick an unassigned vertex~$v$
        \State assign~$V(v)=v$ 
        \State mark~$v$ as free
        \While{there are unassigned vertices adjacent to two
                different assigned vertices}
            \State pick such a vertex~$w$ adjacent to the
                    assigned ~$w_1$ and~$w_2$
            \State assign~$V(w)=V(w_1) \times V(w_2)$
            \State mark edges~$(v,w_1)$ and~$(v,w_2)$ as accounted for
        \EndWhile
    \EndWhile
    \pagebreak[2]
    \For{each pair of vertices~$(v_1, v_2)$}
        \If{$(v_1,v_2)$ is not an edge}
            \State record requirement:~``$V(v_1)$
                    is not collinear to $V(v_2)$''
        \EndIf
    \EndFor
    \For{each edge~$(v_1,v_2)$ not accounted for}
        \State record requirement:~``$V(v_1)$ is orthogonal to~$V(v_2)$''
    \EndFor
\end{algorithmic}

At two points in the algorithm, there is a choice which vertex to pick.
Depending on the vertices chosen, the number of recorded requirements
and free points may significantly vary. By considering all possible choices,
one can find the one with least free points.

The requirements can be mechanically converted
to a formal sentence
in the language of the real numbers.
This sentence is true if and only if the graph is embeddable.
Famously, Tarski proved\cite{tarski}
that such sentences are decidable.
His decision procedure has an impractical complexity.
However, its practical value has been improved
by, for instance, the method of cylindrical algebraic decomposition\cite{qecad}.
We have used the redlog\cite{redlog} package of the reduce algebra
system, which implements a variant of Tarski's quantifier
elimination.\footnote{%
    The reader can find the reduce script generated mechanically
    for the graph in Figure~\ref{fig:unemb-10-2} here:
    \url{http://kochen-specker.info/smallGraphs/49743f49514769444f.html}.  }

Different assignments give different sentences.  In our tests,
some assignments would yield sentences that were decided within milliseconds,
whereas another assignment with less free vertices would
yield a sentence that could not be decided (directly).
Therefore, when determining embeddability of a graph,
we try several assignments in parallel.

In this way, there were still a few (010-colorable) graphs of which
we could not decide embeddability.
With some guessing,
we determined embeddings for these graphs by hand.
Once we knew the troublesome graphs were embeddable,
we adapted the algorithm, as to guess
for some assignments the position of one of the vectors.
If the corresponding
sentence turns out false, we know nothing.  However,
if the sentence is true, we know the graph is embeddable.

With this method, we have decided in a day the embeddability
of every squarefree graph with minimal vertex order three of
less than 15, except for one.\footnote{
A list of all squarefree graphs with minimal vertex order three of
less than 15 vertices together with their embeddability can
be found here:
\url{http://kochen-specker.info/smallGraphs/}.
The graph for which we could not determine
embeddability can be found
here: \url{http://kochen-specker.info/smallGraphs/4d4b3f4b3f603f47414641654953625f3f.html}.
}
In particular:
\begin{comp}\label{comp:unemb20}
    Every squarefree graph of minimal vertex order three
    that is not 010-colorable
    of order less than or equal to 20
    contains, as a subgraph, one of the following three graphs:
    
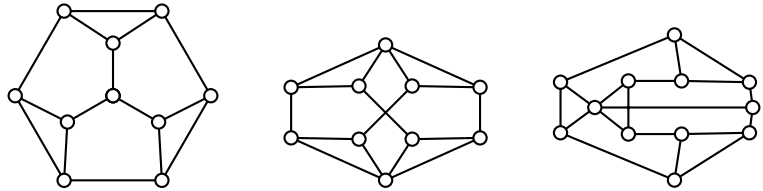
\begin{center}
    \begin{tikzpicture}[thick,scale=1.0,
        p/.style={circle, draw, fill=white,
                        inner sep=0pt, minimum width=5pt}]
            \draw
                \foreach \x in {90,210,330} {
                    (\x:0.7) node[p]{} -- (0:0) node[p]{}
                    (\x-30:1.3) node[p]{} -- (\x:0.7)
                    (\x+30:1.3) node[p]{} -- (\x:0.7)
                    (\x-30:1.3) -- (\x+30:1.3) 
                    (\x+30:1.3) -- (\x-30+120:1.3) 
                }
            ;
        \end{tikzpicture}
        \qquad
    \begin{tikzpicture}[thick,scale=1.0,
        p/.style={circle, draw, fill=white,
                        inner sep=0pt, minimum width=5pt}]
            \draw
                \foreach \x in {0,180} {
                    (\x+15:1.3) node[p]{} -- (\x+90:0.9) node[p]{}
                    (\x+15:1.3) -- (\x+45:0.5) node[p]{}
                    (\x+45:0.5) -- (\x+90:0.9)

                    (\x-15:1.3) node[p]{} -- (\x-90:0.9)
                    (\x-15:1.3) -- (\x-45:0.5) node[p]{}
                    (\x-45:0.5) -- (\x-90:0.9)

                    (\x-15:1.3) -- (\x+15:1.3)
                }
                (45:0.5) -- (45+180:0.5)
                (45+90:0.5) -- (45+90+180:0.5)

            ;
        \end{tikzpicture}
        \qquad
    \begin{tikzpicture}[thick,scale=1.0,
        p/.style={circle, draw, fill=white,
                        inner sep=0pt, minimum width=5pt}]
            \draw
                (45:0.5) node[p]{} -- (135:0.5) node[p]{}
                                -- (225:0.5) node[p]{}
                                -- (315:0.5) node[p]{}
                (15:1.3) node[p]{}
                (345:1.3) node[p]{}
                (195:1.3) node[p]{}
                (165:1.3) node[p]{}

                (75:1) node[p]{} -- (15:1.3) -- (45:0.5) -- (75:1)
                (285:1) node[p]{} -- (345:1.3) -- (315:0.5) -- (285:1)

                (0:1.3) node[p]{} -- (180:0.8) node[p]{}
                (0:1.3) -- (15:1.3)
                (0:1.3) -- (345:1.3)

                (180:0.8) -- (195:1.3)
                (180:0.8) -- (165:1.3)
                (195:1.3) -- (165:1.3)
                (180:0.8) -- (225:0.5)
                (180:0.8) -- (135:0.5)
                (195:1.3) -- (285:1)
                (165:1.3) -- (75:1)

            ;
        \end{tikzpicture}
    \end{center}
    These three graphs are unembeddable.  The left and middle graph
    are the only minimal unembeddable squarefree graph.
For the first graph, we have proven directly that it is unembeddable.
See Proposition~\ref{prop:unemb-10-2}.
For the second graph, we also have a similar direct proof. The third graph is shown to not be embeddable using our algorithm.

Every squarefree graph of minimal vertex order three that is
not~$010$-colorable of order 21 contains an unembeddable subgraph.%
\footnote{A list of these graphs together with their unembeddable
subgraphs, can be found here:
\url{http://kochen-specker.info/candidates/}.
The source code for this computation can be found
at~\texttt{code/comp2.py} of~\cite{GH}.
}
\end{comp}

\section{Conclusion and future research}
Arends, Ouaknine and Wampler struggled with two problems:
enumerating candidate graphs of less than 31 vertices
and testing their embeddability.
We have verified most of their computations.
Then we enumerated all candidate graphs
up to and including 21 vertices.
Furthermore, we have proposed a new decision procedure,
which was able to decide embeddability
for all candidate graphs we found.
Therefore, we demonstrate: a Kochen-Specker system must have at least
22 points.\footnote{%
The authors have a wager whether there is a minimal KS system of less
than 25 vertices.}

Enumerating all candidate graphs of less than 31 vertices
is computationally infeasable.
To bridge the enormous the gap between 22 and 31,
requires a new insight.
For instance: another restriction on which graphs to consider.

The Reader, interested in pursuing this line of research,
is encouraged to read the master thesis\cite{a09} of Arends,
in which he discusses in detail several other
properties that a minimal KS system must enjoy, as well as
some failed attempts.

\section{Acknowledgments}
We wish to thank the following for their generous contribution to the
distributed computation:
    the Digital Security group, Intelligent Systems group
    and the C\&CZ service of the Radboud University;
    Wouter Geraedts and
    Jille Timmermans.

We are grateful to prof.~McKay for discussing
the feasibility of certain graph restrictions and
to Judith van Stegeren\footnote{\url{http://jd7h.com/}}
for drawing figures for the draft version of this paper.

\clearpage
\bibliography{main}{}
\bibliographystyle{eptcs}

\end{document}